\documentclass[a4paper,USenglish,cleveref]{lipics-v2021}

\pdfoutput=1
\hideLIPIcs
\nolinenumbers
\title{Knowledge in multi-robot systems: an interplay of dynamics, computation and communication}
\titlerunning{Knowledge in multi-robot systems}

\author{Giorgio Cignarale}{TU Wien, Vienna, Austria}{giorgio.cignarale@logic.at}{0000-0002-6779-4023}{}
\author{Stephan Felber}{TU Wien, Vienna, Austria}{stephan.felber@tuwien.ac.at}{0009-0003-6576-1468}{}
\author{Eric Goubault}{École Polytechnique, Palaiseau, France}{eric.goubault@lix.polytechnique.fr}{0000-0002-3198-1863}{}
\author{Bernardo Hummes Flores}{École Polytechnique, Palaiseau, France}{bernardo.hummes-flores@lix.polytechnique.fr}{0000-0003-2325-1497}{}
\author{Hugo Rincon Galeana}{TU Berlin, Berlin, Germany}{hugorincongaleana@gmail.com}{0000-0002-8152-1275}{}

\authorrunning{G. Cignarale, S. Felber, E. Goubault, B. Hummes Flores, H. Rincon Galeana}



\usepackage{amsfonts,amsmath,amsthm,amssymb,bm}
\usepackage{graphicx}
\usepackage{tikz,pgf}
\usetikzlibrary{arrows.meta}
\usetikzlibrary{cd}
\usetikzlibrary{fit,shapes.geometric}
\usetikzlibrary{backgrounds}
\usepackage{environ}
\usepackage{mfirstuc}
\usepackage{cite}
\usepackage{algorithmic}
\usepackage{textcomp}
\usepackage{xcolor}
\usepackage{booktabs}

\newtheorem{assumption}[theorem]{Assumption}
\Crefname{assumption}{assumption}{assumptions}

\renewcommand{\geq}{\geqslant}

\renewcommand{\epsilon}{\varepsilon}
\renewcommand{\emptyset}{\varnothing}

\newcommand{\N}{\mbox{$\mathbb{N}$}}

\newcommand{\R}{\mbox{$\mathbb{R}$}}

\newcommand{\floors}[1]{\left\lfloor #1 \right\rfloor}

\renewcommand{\hat}[1]{\widehat{#1}}

\newcommand{\demph}[1]{\textbf{#1}} % Definitional emphasis
\newcommand{\remph}[1]{{#1}} % {\emph{#1}} Rhetorical emphasis

\newcommand{\syslcm}{\ensuremath{\mathcal{C}}} % computational system
\newcommand{\sysdyn}{\ensuremath{\mathcal{H}}} % hybrid system
\newcommand{\disc}{\ensuremath{\Phi}}          % discretization map

\newcommand{\onc}{\ensuremath{\mathbf{x}}}  % ontic state
\newcommand{\epi}{\ensuremath{\mathbf{e}}}  % epistemic state
\newcommand{\fonc}{\ensuremath{f}}          % ontic evolution
\newcommand{\fepi}{\ensuremath{\varphi}}    % epistemic evolution

\newcommand{\obs}{\ensuremath{\mathbf{y}}}  % observation
\newcommand{\ctl}{\ensuremath{\mathbf{u}}}  % control
\newcommand{\gobs}{\ensuremath{g}}          % observation generation
\newcommand{\gctl}{\ensuremath{h}}          % control generation

\newcommand{\ONC}{\ensuremath{\mathbf{X}}}  % distributed ontic state
\newcommand{\EPI}{\ensuremath{\mathbf{E}}}  % distributed epistemic state
\newcommand{\FONC}{\ensuremath{F}}          % distributed ontic evolution
\newcommand{\FEPI}{\ensuremath{\Phi}}       % distributed epistemic evolution

\newcommand{\OBS}{\ensuremath{\mathbf{Y}}}  % distributed observation
\newcommand{\CTL}{\ensuremath{\mathbf{U}}}  % distributed control
\newcommand{\GOBS}{\ensuremath{G}}          % distributed observation function 
\newcommand{\GCTL}{\ensuremath{H}}          % distributed control function

\newcommand{\per}{\ensuremath{\hat{\onc}}}  % perception
\newcommand{\com}{\ensuremath{\hat{\epi}}}  % communication
\newcommand{\gper}{\ensuremath{\eta}}       % perception function
\newcommand{\gcom}{\ensuremath{\lambda}}    % communication function

\newcommand{\OBLOT}{\ensuremath{\mathcal{OBLOT}}}  % oblivious robot model
\newcommand{\LUMI}{\ensuremath{\mathcal{LUMI}}}    % luminous robot model
\newcommand{\LCM}{\textsc{LCM}}
\newcommand{\look}{\textsc{look}}
\newcommand{\compute}{\textsc{compute}}
\newcommand{\move}{\textsc{move}}

\newcommand{\robot}{\ensuremath{\mathbf{r}}}    % robot state
\newcommand{\robots}{\ensuremath{\Pi}}          % set of robots
\newcommand{\env}{\ensuremath{\mathcal{E}}}     % environment state
\newcommand{\adv}{\ensuremath{\mathcal{B}}}     % adversary
\newcommand{\tpath}{\ensuremath{p}}             % time path
\newcommand{\run}{\ensuremath{\sigma}}          % system run
\newcommand{\runs}{\ensuremath{\Sigma}}         % system run
\newcommand{\algo}{\ensuremath{\mathcal{A}}}    % algorithm
\newcommand{\cfg}{\ensuremath{\mathcal{C}}}     % global configuration

\newcommand{\propspace}{\ensuremath{\text{P}_{X}}}
\newcommand{\propcyl}{\ensuremath{\text{P}_{X\times[0,1]}}}

\newcommand{\xp}{\ensuremath{\text{sp}}}

\newcommand{\found}{\ensuremath{\text{FOUND}}}
\newcommand{\secure}{\ensuremath{\text{SECURE}}}
\newcommand{\propsurv}{\ensuremath{\text{P}_{Surv(X)}}}
\newcommand{\trace}{\ensuremath{\text{tr}}}

\newcommand{\propobs}{\ensuremath{\text{P}_{g(X)}}}
\newcommand{\proptopo}{\ensuremath{\text{P}_{\tau}}}
\newcommand{\proprdvz}{\ensuremath{\text{P}_{rz(X)}}}

\begin{document}
\maketitle

\begin{abstract}
  In this paper, we provide a framework integrating distributed multi-robot systems and temporal epistemic logic. We show that continuous-discrete hybrid systems are compatible with logical models of knowledge already used in distributed computing, and demonstrate its usefulness by deriving sufficient epistemic conditions for exploration and gathering robot tasks to be solvable. We provide a separation of the physical and computational aspects of a robotic system, allowing us to decouple the problems related to each and directly use methods from control theory and distributed computing, fields that are traditionally distant in the literature. Finally, we demonstrate a novel approach for reasoning about the knowledge in multi-robot systems through a principled method of converting a switched hybrid dynamical system into a temporal-epistemic logic model, passing through an abstract state machine representation. This creates space for methods and results to be exchanged across the fields of control theory, distributed computing and temporal-epistemic logic, while reasoning about multi-robot systems.
\end{abstract}

\section{Introduction} % {{{
\label{sec:Intro}
% Intro {{{
We formulate in this paper a framework connecting a hybrid systems approach to distributed multi-robot systems and temporal epistemic logic. 
This forms a bridge whose value is twofold: we bring the success of logical analysis of tasks in distributed computing \cite{bookof4} to distributed robotics, and we generalize the multi-robot model from state machines to dynamical systems \cite{PostoyanFrascaPanteleyZaccarian24}, drawing a much closer frontier to real-life applications.

The distributed coordination of robots combines problems found in control theory \cite{robot-control,khalil} with the problems commonly found in distributed computing \cite{lynch}.
By clearly separating the roles of the physical and computational parts of a robot, we are capable of showing a translation layer from distributed computing and control theory.
While robotics is traditionally concerned with physical entities that must behave in accordance to their dynamical constraints, the distributed computing perspective arises once the robots must exchange information under varying difficulties associated to their communication capacities \cite{AlcantaraCastanedaFlores-PenalozaRajsbaum19,SN17}. For instance, understanding the impact of synchrony, coordination and message loss are essential for the efficient design and analysis of state-of-the-art multi-robot systems solving various tasks. 

Our modeling is extended for the multi-robot scenario, where we introduce the notion of a scheduler as a path in the space of executions, so to tie local to global time.
Bridging these fields enables us to combine well-known methods and results from control theory, distributed computing and temporal epistemic logic in a simple and harmonious fashion. 
Our framework successfully captures the spatial information used to reason about multi-robot tasks, namely exploration and gathering, as it is constructed with the memory of the robots during their computations.
We devise a logical model consisting of a variation of the runs and systems framework, which we construct by defining explicitly a distributed state machine for both the system dynamics and the robot protocols.
We further demonstrate the usefulness of our approach by deriving the first ever sufficient epistemic conditions for solving the exploration, surveillance and approximate gathering robot tasks.

In \cite{AlcantaraCastanedaFlores-PenalozaRajsbaum19} the authors famously showed that exact gathering is impossible by reducing the robot task to the consensus problem in distributed computing and deriving its well understood impossibility borders from there. In essence, the present paper puts this reduction on a sound footing by deriving the distributed system implied by the control theoretic model and equipping it with the tools found in epistemic logic.

The distributed computing community has approached the characterization of tasks in distributed robot systems \cite{FlocchiniPrencipeSantoro19} in various ways, such as according to the system space (continuous or discrete)\cite{BalabonskiCourtieuPelleRiegTixeuilUrbain19,DEmidioDiStefanoFrigioniNavarra18}, communication capacities\cite{TeraiWadaKatayama23} and synchronicity\cite{KirkpatrickKostitsynaNavarraPrencipeSantoro24}. Model checking using linear temporal logic is also a common tool in this computational perspective of robotics, where it has been used for gathering \cite{DoanBonnetOgata18}, perpetual exploration \cite{DoanBonnetOgata17,BerardLafourcadeMilletPotop-ButucaruThierry-MiegTixeuil16} and terminating exploration \cite{BerardLafourcadeMilletPotop-ButucaruThierry-MiegTixeuil16}.

The control theoretic view on multi-agent systems has been developed by a number of authors, see e.g. the survey \cite{Heemels}. Similarly, logical views on hybrid systems have been considered, see e.g. \cite{Platzer}. As we mentioned also, logics, in particular epistemic logic \cite{HM90} and temporal-epistemic logic \cite{Knight} have been developed and applied to multi-agent systems, through in particular epistemic planning \cite{FabianoSrivastavaGanapiniLenchnerHoreshRossi}.

Epistemic logic \cite{Hin62} is a powerful conceptual tool for reasoning about the uncertainty of agents in a system, both regarding facts and other agents' epistemic states, making it especially suitable for the modeling of multi-agent systems, including distributed systems. It was adapted for and introduced to distributed systems in \cite{HM90}, subsequently employed to characterize: the firing rebels with relay in \cite{FireTark}, two agent approximate consensus in \cite{LRA20}, sufficient conditions for stabilizing consensus in \cite{crgfessli2024}, optimal consensus with respect to message transmission in\cite{GM20} and different notions of knowledge in \cite{M16}, to name just a few. A crucial feature of epistemic logic in multi-agent systems is that it provides a precise and formal representation of higher order epistemic attitudes, such as nested and/or common knowledge, often crucial for the success of distributed tasks. For instance, coordinated actions are inevitably tied to nested knowledge, i.e., any agent needs to know to some extent what other agents know in order to successfully execute a coordinated behavior~\cite{M16}. Another interesting feature of epistemic logic is that, when combined with temporal logic, it allows to express properties that must hold (and be known) throughout the execution, without focusing on the specific communication mechanisms leading to such knowledge. Thus, epistemic logic is a universal vehicle for expressing communication models, including those that rely on lack of explicit communication (also known as communication by silence) \cite{FireTark,GM20}.
Our approach is the first one to seamlessly integrate these three complementary views: the ones from control theory, distributed systems and temporal epistemic logic.
% }}}

\textbf{Contributions.} % {{{
We present, to the extent of our knowledge, the first interdisciplinary framework that connects all three fields of control theory, distributed computing, and temporal epistemic logic, in a general setting that enables a comprehensive and integral approach to multi-robot systems. We showcase the potential of our framework by defining a simple yet sufficiently expressive temporal epistemic language that allows us to express multi-layered knowledge of space. In particular, this enables a precise formalization for the completion of spatial robotic tasks, such as exploration and approximate gathering. More precisely, our main contributions are:

\begin{itemize}
    \item We define a robot abstraction representing multi-robot system (\Cref{def:multi-robot-dynamics}) executions compatible with to both the control theory (\Cref{def:robot-ode}) and distributed computing (\Cref{def:robot-state-machine,def:env-state-machine}) frameworks.
    \item We provide a temporal epistemic logic framework (\Cref{def:epi-frame}), based on the runs and systems framework, for expressing robot exploration as knowledge of atomic propositions that represent open sets of a compact topological space.
    \item We provide an epistemic analysis of exploration (\Cref{def:epi_cond}), and approximate gathering (\Cref{def:term_approx_gathering}). Furthermore, we characterize exploration with cooperative termination (\Cref{thm:termequiv}), and we prove that the sufficient conditions to solve approximate gathering mimic the sufficient conditions for the stabilizing agreement problem in distributed systems (\Cref{thm:stabagapproxgather}).
\end{itemize}

% contributions }}}

\textbf{Paper Organization.} % {{{
In \Cref{sec:robot-model} we introduce the basic concepts used for modeling multiple robots: first as a system of ordinary differential equations (\Cref{sec:robot-as-dynsys}), then as state machines (\Cref{sec:robot-as-computer}). In \Cref{sec:robot-executions}, we introduce time paths as a geometric model relating global and local times (\Cref{sec:scheduler}), followed by their usage in the definitions of runs in both dynamical (\Cref{sec:runs-dyn}) and computational models (\Cref{sec:runs-lcm}). We conclude by showing that the computational content of these runs match (\Cref{sec:correspondence}). \Cref{sec:task-tel}, the main contribution of the paper, connects the robotic abstraction to the epistemic modeling by turning the abstract robot state machine into a suitable epistemic frame. The latter is enriched into an epistemic model (\Cref{sec:logexplore}), which is then used for the epistemic analysis of robot tasks such as simple exploration (\Cref{sec:LCMexplore}) and approximate gathering (\Cref{sec:gathering}). Finally, concluding thoughts are provided in \Cref{sec:concl}.
% organization }}}

% section }}}

\section{The Robot Model} % {{{
\label{sec:robot-model}

In this section, we introduce two multi-robot system models with distinct degrees of abstractions, and subsequently show one model to be an abstraction of the other.
The first is a description of robots as a switched hybrid dynamical system, close to models found in the control theory literature \cite{Jaulin19}. The second is a state machine executing \textsc{look-compute-move} cycles, an abstraction commonly found in the distributed computing community \cite{FlocchiniPrencipeSantoro19}.

\subsection{Robots as Dynamical Systems} % {{{
\label{sec:robot-as-dynsys}

% one robot {{{
From the dynamical systems perspective, a robotic system can modeled via differential equations that describe the evolution of its state over time. Hybrid systems combine a continuous state space, associated to the physical features of a robot, with a discrete space, more natural to capture computations. We start by describing the model for a single robot.

\begin{definition}[Robot]\label{def:robot-ode}
  A \demph{robot} is a hybrid system, composed of an \demph{ontic state} \(\onc\) and \demph{epistemic state} \(\epi\). It has actuators that evolve its ontic state and sensors that perform observations. It has an intelligence and a memory that generate new epistemic states and control commands.
  \begin{align}
    \dot{\onc}(t) &= \fonc(\onc(t),\ctl[t]) &&\text{(ontic evolution)} \label{eq:f-ontic}\\
    \obs(t) &= \gobs(\onc(t)) &&\text{(observation)} \label{eq:g-observation}\\
    \epi[t+1]  &= \fepi(\epi[t],\obs[t]) &&\text{(epistemic evolution)} \label{eq:f-epistemic}\\
    \ctl[t] &= \gctl(\epi[t],[t]) &&\text{(control)} \label{eq:g-control}
  \end{align}
  \noindent The ontic evolution $\fonc$ and the observation $\gobs$ occur in continuous time \(t\), while the epistemic evolution $\fepi$ and the control $\gctl$ occur in discrete time \([t]\), defined by the staircase function \(\floors{t}\).
\end{definition}

The ontic state corresponds to the physical state of the robot, while the epistemic state corresponds to the memory state of the robot. The distinction between the ontic and epistemic states allows for physical and computational problems to be described separately. As such, a robot's behavior is given by a (hybrid) switched system \cite{Liberzon03}, a particular case of hybrid system, as the control $\ctl[t]$ is possibly discontinuous over time. 

Note that \Cref{def:robot-ode} also encompasses more traditional models of robots in the control-theoretic literature, such as output-feedback controllers \cite{khalil}. This can be achieved by using the epistemic evolution \(\fepi\) as a simulator that imitates the dynamics of the model. 

\begin{remark}[Discrete Observations]\label{rem:continuoustodiscrete}
  Due to the physical limitations in memory of any physical computing device, we assume that the observation function \(\gobs\) has a finite domain \(\mathcal{D}(g)\), which corresponds to the digital representation of the continuous quantities it is capable of measuring.
\end{remark}
% end one robot }}}

% multiple robots % {{{
We now consider the modifications required for modeling a multi-robot system where interactions increase in complexity with the added notions of \remph{communication} and \remph{perception} among robots. Note that the distinction between communication and perception is a consequence of the distinction between ontic and epistemic states, as both are required to represent the possible information to be exchanged. For two robots \(\robot\) and \(\robot'\), they are defined via a pair of functions.
\begin{align}
  \per_{\robot \robot'}(t) &= \gper_{\robot}(\onc_{\robot'}(t)) &&\text{(perception)} \label{eq:per}\\
  \com_{\robot \robot'}[t] &= \gcom_{\robot}(\epi_{\robot'}[t]) &&\text{(communication)} \label{eq:com}
\end{align}

The perception is \(\per_{\robot,\robot'}\) is the estimation of the ontic state of robot \(\robot'\) by robot \(\robot\). Similarly, The communication function generates the estimation \(\com_{\robot,\robot'}\) of the epistemic state of robot \(\robot'\) by robot \(\robot\). Moreover, the introduction of multiple robots operating in the system brings forward the notions of "local" and "global" execution time, which will be formalized in \Cref{sec:scheduler}.

Communication and perception complement the observation step with information about the other robots. These multiple pieces of information are later processed together by the robot's epistemic evolution, matching the intuition of observations as simply inputs to a computing system. This leads to the definition of a hybrid multi-robot system.

\begin{definition}[Hybrid multi-robot system]\label{def:multi-robot-dynamics}
  A \demph{hybrid multi-robot system} \(\sysdyn\) consists of \(N\) robots, capable of observing the environment and interacting with the other robots through perception and communication. For a given sequence of activations $\tpath$, its evolution is given by:
    \begin{align}
        \dot{\ONC}(t) &= \FONC_{\tpath}(\ONC(t),\CTL[t]) &&\text{(global ontic evolution)} \label{eq:gbl-f-ontic}\\
        \OBS(t) &= \GOBS(\ONC(t), \gper(\ONC(t)),\gcom(\EPI[t])) &&\text{(multi-robot observation)} \label{eq:gbl-g-multi-observation}\\
        \EPI[t+1]  &= \FEPI(\EPI[t],\OBS(t)) &&\text{(global epistemic evolution)} \label{eq:gbl-f-epistemic}\\
        \CTL[t] &= \GCTL(\EPI[t],t) &&\text{(global control)} \label{eq:gbl-g-control}
    \end{align}
\end{definition}

This definition is a consequence of the evolution of single robots alongside the interaction functions, as in \Cref{def:robot-ode,eq:per,eq:com}. We now consider a generic scheduler that provides the set of active robots at each time \(t\), defining the modes of the switched system, and that will be made precise with the introduction of time paths in \Cref{sec:scheduler}.
% end multiple robots }}}
% end robots as dynamical systems }}}

\subsection{Robots as Mobile Computing Systems} % {{{
\label{sec:robot-as-computer}

% LCM % {{{
We provide now a state machine abstraction that allows us to focus on the computational component of the robots in \Cref{def:robot-ode}. Robots are algorithmically described as mobile entities that follow a cycle of \textsc{look-compute-move} operations as they interact with the environment. This is known as the \LCM\ cycle\cite{FlocchiniPrencipeSantoro19}:

\begin{enumerate}
  \item \look. The robot retrieves a snapshot of the ontic state of other robots in the environment.
  \item \compute. The robot computes its next epistemic state from the information in \look.
  \item \move. The robot updates its physical state based on its current epistemic state.
  \item \textsc{wait}. The robot does not perform any action, but remains observable.
\end{enumerate}

All robots operate in the same environment, e.g. a graph or continuous region of space. Uncertainty in the behavior of each step, such as the accuracy of the information or the actual moved distance, are possible choices for adversarial behavior in the \LCM\ model. Some variations are depicted in Appendix~\ref{apx:lcm}. 

Two models for robots running \LCM\ cycles are predominant in the literature: \remph{oblivious}, introduced in \cite{SuzukiYamashita99}, and \remph{luminous}, introduced in \cite{DasFlocchiniPrencipeSantoroYamashita16}. They contrast in the means for communication and memory made available to the robots. Robots under the oblivious robot model, \OBLOT, have no memory and can only obtain information from observing other robots in the \look\ step. Robots under the luminous robot model, \LUMI, are enhanced with a set of programmable lights that can be user for both (indirect) communication and storing information for future \compute\ steps. The lights are a physical element representing the exposed memory of a robot, and hence can be observed with the rest of the ontic state of a robot during \look\ operations and are updated during a \move\ operation. We will use the \LUMI\ model, as we require the lights to handle both communication and perception. We instantiate the robot \(\robot\) in an environment \(\env\) executing \(\algo\) as follows.

% LCM }}}

% State machines{{{
\begin{definition}[Robot state machine] \label{def:robot-state-machine}
  A \remph{robot} \(\robot\) executes the state machine\\ 
  \(\algo = \langle S_{\epi_{\robot}}, S_{\obs_{\robot}}, S_{\ctl_{\robot}}, \gobs_{\robot}, \fepi_{\robot}, \gctl_{\robot}\rangle\)
  where

  \begin{itemize}
    \item \(S_{\epi_{\robot}}\) is a finite set of possible memory states;
    \item \(S_{\obs_{\robot}}\) is a finite set of possible observations obtained from the environment \(\env\);
    \item \(S_{\ctl_{\robot}}\) is a set of actions that can be executed in the environment \(\env\);
    \item \(\gobs: \env \to S_{\obs_{\robot}}\) is the function that retrieves an observation from the current environment \(\env\);
    \item \(\fepi: S_{\epi_{\robot}} \times S_{\obs_{\robot}} \to S_{\epi_{\robot}}\) is a function taking an epistemic state and an observation to a new epistemic state;
    \item \(\gctl: S_{\epi_{\robot}} \to S_{\ctl_{\robot}}\) is a function that generates an action given an epistemic state.
  \end{itemize}
\end{definition}

The environment itself is dynamic and therefore also a state machine. It provides the observations used by the robots during their execution of \(\algo\) and evolves according to the robots' actions and some possible adversarial factor. It abstracts the ontic state of all robots as a single entity, with which all of the robot state machines interact.

\begin{definition}[Environment state machine] \label{def:env-state-machine}
  The \demph{environment} executes the state machine
  \(\env = \langle S_{\env}, I_{\adv}, \{S_{\ctl_{\robot}}\}_{\robot \in \Pi},\{S_{\obs_{\robot}}\}_{\robot \in \Pi}, \FONC_\env, \GOBS_\env \rangle\)
  where
  
  \begin{itemize}
    \item \(S_{\env}\) is a set of possible environment states;
    \item \(I_{\adv}\) is a set of possible choices made by the adversary;
    \item \(\{S_{\ctl_{\robot}}\}_{\robot \in \Pi}\) is a collection of possible robot actions;
    \item \(\{S_{\obs_{\robot}}\}_{\robot \in \Pi}\) is a collection of possible observations;
    \item \(\FONC_\env: S_{\env} \times \prod_{\robot \in \Pi} S_{\ctl_{\robot}} \times I_{\adv} \to S_{\env}\) is the function generating a new environment state from the latest actions and the adversarial interference;
    \item \(\GOBS_\env: S_{\env} \times I_{\adv} \to \prod_{\robot \in \Pi} S_{\obs_{\robot}}\) is the function that generates observations from the current environment state and possible adversarial interference.
  \end{itemize}
\end{definition}

\Cref{def:robot-state-machine} lets us focus on the computations executed locally, while \Cref{def:env-state-machine} abstracts the ontic states and updates of all robots in a single machine interacting with the individual computations. \look\ and \move\ are split into distinct functions at the interface between the environment and each robot: the data is ``prepared'' by \(\gctl_\robot\) and \(\GOBS_\env\), and ``consumed'' by \(\FONC_\env\) and \(\gobs_\robot\).

A \demph{computational multi-robot system} \(\syslcm\) is then a set of robots and one environment state machine, i.e., \(\syslcm = (\{\algo_{\robot}\}_{\robot \in \robots}, \env)\). This separation will facilitate stating the correspondence between computational and dynamical perspectives in \Cref{sec:correspondence}. 
% }}}

% end Robots as mobile computing systems }}}

% end The Robot Model }}}

\section{Schedulers and System Runs} % {{{
\label{sec:robot-executions}

In this section, we define a \remph{scheduler} (\Cref{sec:scheduler}) that dictates the possible activations of the distributed system. \remph{System runs} are thus introduced, in both the dynamical (\Cref{sec:runs-dyn}) and computational (\Cref{sec:runs-lcm}) definitions. The runs will be later used as building blocks to generate a general epistemic frame for multi-robot tasks in \Cref{sec:task-tel}.

\subsection{Scheduler: Global and Local Time} % {{{
\label{sec:scheduler}

The robots' order of activation is given by a \remph{scheduler}, which chooses the robot that will execute the next step of its cycle. Common schedulers found in the literature are the synchronous, where all steps happens simultaneously, semi-synchronous, where the synchronization is guaranteed at every full cycle, and asynchronous, where there is no shared notion of time. See \cite{FlocchiniPrencipeSantoro19} for more details on the taxonomy of schedulers. A scheduler that is not fully synchronous may lead to a global time that differs from the local time of the robots. 

We introduce a \remph{time path} to capture the relationship between the \remph{local} execution of the robots and a common \remph{global} time variable. A set of \(n\) robots induces a space of possible executions \(\R_{+}^{n}\), \LCM\ cycles are mapped to each sequential unit interval. Any monotonic increasing path rooted in the origin of this space defines a possible execution order of the system. Each coordinate of \(\R_{+}^{n}\) corresponds to the \remph{local time} of a robot, and the length of the path corresponds to the \remph{global time}. A robot is said to execute a step whenever a path's projection to the corresponding local time axis crosses said step's activation point. This leads to the following definition.

\begin{definition}[Time path]\label{def:time}
  For a system of \(n\) robots, a \demph{time path} is a differentiable map \(\tpath\), from \([0, T]\) to \(\R^{n}_{+}\), for some \(T \ge 0\), that has the following properties:
  \begin{itemize}
    \item \textbf{(initialisation)} \(\tpath(0) = (0,\dots,0)\).

    \item \textbf{(monotonicity)} For any \(t, t' \in [0,T]\) with \(t \le t' , \tpath_i(t) \le \tpath_i(t')\), for all components \(i \in [1,\dots,n]\).

    \item \textbf{(rectification)} For all \(t \in [0,T]\), the global time \(t\) is the maximum of all local times, i.e., \(t = max(|\tpath_{1}(t)-\tpath_{1}(0)|,|\tpath_{2}(t)-\tpath_{2}(0)|,\dots,|\tpath_{n}(t)~-~\tpath_{n}(0)|)\).
  \end{itemize}
\end{definition}

A time path defines the relative progress of every robot, where the monotonicity property makes time always go forwards, and the rectification property sets the global time to match the fastest local time. A \demph{scheduler} is then a family of time paths. For instance, a synchronous scheduler can only correspond to a diagonal straight line, where all operations are aligned. A semi-synchronous scheduler will have any of its possible time paths crossing all \compute\ steps simultaneously, where synchronization happens. In the case of instantaneous moves, as in \cite{AlcantaraCastanedaFlores-PenalozaRajsbaum19}, the behavior of robots would be undefined at points where a \look\ operation would happen at the same time as a \move\ operation, and consequently one should forbid access to points near the intersection of \move\ and \look\ operations. Such forbidden zones (represented as small squares) and two possible executions of a multi-robot system composed of two robots is illustrated in \Cref{fig:time-path}.

\begin{figure}[ht]
  \centering
  \begin{tikzpicture}[scale=.60]
  \node [left] at (0,7) {$[t]_2$};
  \node [below] at (13,0) {$[t]_1$};

  \node [left] at (0,1) {$\textsc{m}_0$};
  \node [left] at (0,2) {$\textsc{l}_0$};
  \node [left] at (0,3) {$\textsc{c}_0$};
  \node [left] at (0,4) {$\textsc{m}_1$};
  \node [left] at (0,5) {$\textsc{l}_1$};
  \node [left] at (0,6) {$\textsc{c}_1$};

  \node [below] at (1,0) {$\textsc{m}_0$};
  \node [below] at (2,0) {$\textsc{l}_0$};
  \node [below] at (3,0) {$\textsc{c}_0$};
  \node [below] at (4,0) {$\textsc{m}_1$};
  \node [below] at (5,0) {$\textsc{l}_1$};
  \node [below] at (6,0) {$\textsc{c}_1$};
  \node [below] at (7,0) {$\textsc{m}_2$};
  \node [below] at (8,0) {$\textsc{l}_2$};
  \node [below] at (9,0) {$\textsc{c}_2$};
  \node [below] at (10,0) {$\textsc{m}_3$};
  \node [below] at (11,0) {$\textsc{l}_3$};
  \node [below] at (12,0) {$\textsc{c}_3$};

  % move t1
  \draw[dash dot, cyan] (3,0) -- (3,6);
  \draw[dash dot, cyan] (6,0) -- (6,6);
  \draw[dash dot, cyan] (9,0) -- (9,6);
  \draw[dash dot, cyan] (12,0) -- (12,6);

  % move t2
  \draw[dash dot, cyan] (0,3) -- (12,3);
  \draw[dash dot, cyan] (0,6) -- (12,6);

  % look t1
  \draw[dotted,thin,red] (1,0) -- (1,6);
  \draw[dotted,thin,red] (4,0) -- (4,6);
  \draw[dotted,thin,red] (7,0) -- (7,6);
  \draw[dotted,thin,red] (10,0) -- (10,6);

  % compute t1
  \draw[dotted,thin,red] (2,0) -- (2,6);
  \draw[dotted,thin,red] (5,0) -- (5,6);
  \draw[dotted,thin,red] (8,0) -- (8,6);
  \draw[dotted,thin,red] (11,0) -- (11,6);

  % look t2
  \draw[dotted,thin,red] (0,1) -- (12,1);
  \draw[dotted,thin,red] (0,4) -- (12,4);

  % compute t2
  \draw[dotted,thin,red] (0,2) -- (12,2);
  \draw[dotted,thin,red] (0,5) -- (12,5);

  % Update/Scans

  \fill[gray] (0.8,1.8) rectangle  (1.2,2.2);
  \fill[gray] (1.8,0.8) rectangle  (2.2,1.2);
  \fill[gray] (0.8,4.8) rectangle  (1.2,5.2);
  \fill[gray] (1.8,3.8) rectangle  (2.2,4.2);
  \fill[gray] (3.8,4.8) rectangle  (4.2,5.2);
  \fill[gray] (4.8,3.8) rectangle  (5.2,4.2);
  \fill[gray] (3.8,1.8) rectangle  (4.2,2.2);
  \fill[gray] (4.8,0.8) rectangle  (5.2,1.2);
  \fill[gray] (6.8,1.8) rectangle  (7.2,2.2);
  \fill[gray] (7.8,0.8) rectangle  (8.2,1.2);
  \fill[gray] (6.8,4.8) rectangle  (7.2,5.2);
  \fill[gray] (7.8,3.8) rectangle  (8.2,4.2);
  \fill[gray] (9.8,4.8) rectangle  (10.2,5.2);
  \fill[gray] (10.8,3.8) rectangle  (11.2,4.2);
  \fill[gray] (9.8,1.8) rectangle  (10.2,2.2);
  \fill[gray] (10.8,0.8) rectangle  (11.2,1.2);

  \draw[ultra thick,red, rounded corners = 4pt] (0,0) -- (2.5,0.5) -- (3.5,2.5)  -- (5,3) -- (7.5,3.5) -- (8,5) -- (11,6);

  \draw[ultra thick, rounded corners = 4pt, blue] (0,0) -- (0.5,1.3) -- (5.2,1.4) -- (6,4.5) -- (8,4.5) -- (10.8,4.8) -- (12,6);

  \draw[->] (-0.5,0) -- (13,0);
  \draw[->] (0,-0.5) -- (0,7);
  \draw[->, ultra thick, blue] (11.5,5.5) -- (12.2,6.2);
  \draw[->, ultra thick, red] (9.5,5.5) -- (11,6);

  \end{tikzpicture}
  \caption{Two global time paths, blue and red, represent different possible executions of robots with local times \([t]_1\) and \([t]_2\). Each robot executes the sequence of actions (\textsc{M}) update ontic state, (\textsc{L}) observe environment and other robots, (\textsc{C}) update epistemic state. Each moment the time path crosses a dotted line, the associated robot executes the corresponding action. The squares are used to distinguish paths with different activation orders, as a \textsc{look} operation before and after a \textsc{move} provides different amounts of information. The actions are indexed by the cycle they belong to.}
  \label{fig:time-path}
\end{figure}
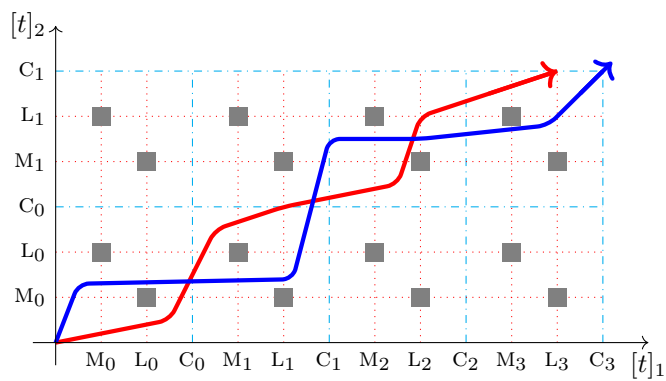

In the following, we use \(t\) for the \remph{continuous} time and \([t]\) for \remph{discrete} time, as per \Cref{def:robot-ode}. A subscript with the robot identifier is used to depict \remph{local} time, as in \(t_\robot\) or \([t]_\robot\), and its absence represents \remph{global} time. Given a time path $\tpath$, the local time of robot $\robot$ is given by \(t_\robot\) or \([t]_\robot\), depending on whether it is continuous or discrete.
% end Scheduler }}}

\subsection{Dynamical Runs of Distributed Multi-Robot System} % {{{
\label{sec:runs-dyn}

We attest now the uniform treatment of time provided by the scheduler by bridging the understanding of system runs of robots as dynamical systems (\Cref{sec:robot-as-dynsys}) similarly to that of computing systems (\Cref{sec:robot-as-computer}).

Given a time path \(\tpath\), the hybrid dynamics in global time coordinates \(\ONC(t) = \onc(\tpath_i(t))\) for every robot \(i\) satisfy 
\[
  \dot{\ONC}(t) = \FONC_{\tpath}(\ONC(t), \CTL[t]) = \langle \dot{\tpath}, \fonc(\ONC(t), \CTL[t]) \rangle
\]
where \(\langle \cdot, \cdot \rangle\) denotes the scalar product in \(\R^{n}\). This connect the scheduler to the physical evolution via the time path.

For simplicity, we assume that all robots live in a hypercube bounded by the $[0,1]$ interval \(X  \subseteq [0,1]^k\), with \(k \ge 1\), along with its standard Euclidean metric, i.e. $d(\Bar{x}, \Bar{y}) = \sqrt{\sum_{i=1}^{k} (x_i - y_i)^2 }$, where $\Bar{x} = (x_1, \ldots,x_k)$ and $\Bar{y}= (y_1, \ldots, y_k)$. We consider only the case of homogeneous multi-robot systems, equipped with identical capacities and behavior, but note that the described model does not limit to such case.

The (arbitrary) switched system of Definition \ref{def:multi-robot-dynamics} has a notion of solution, which will be used later as the dynamical system version of ``runs''. For this, we need to see our switched system as a more general differential inclusion: 

\begin{definition}[Solutions to differential inclusions\cite{Aubin}] \label{def:solutions} 
Consider the general differential inclusion \(\dot{x} \in \mathbb{F}(x)\) where $\mathbb{F}$ is a map from $\R^n$ to $\mathcal{P}(R^n)$, the set of subsets of $\R^n$. A function $x(\cdot): \mathbb{R}^+\rightarrow \mathbb{R}^n$ is a \demph{solution} to the inclusion if $x$ is an absolutely continuous function and satisfies for almost all $t\in\mathbb{R}$, $\dot{x}(t)\in \mathbb{F}(x(t))$.
\end{definition}

In general, there can be many solutions to a differential inclusion. Throughout the section we note $S_{\mathbb{F}}(x_0)$ the set of all (absolutely continuous) solutions to the differential inclusion.

In some cases, we can be more precise about the solution set. By the Filippov-Wa$\breve{\mbox{z}}$ewski theorem \cite{Liberzon03}, all solutions to the (closure of the) convexification of a differential inclusion 
can be approximated by solutions of the original differential inclusion with the same initial value, at least over a compact time interval, and under some simple hypotheses. And the solutions of the closure of the convexification of the original differential inclusion exist under simple hypotheses: 

We know from \cite{Aubin} that when $\mathbb{F}$ is a Marchaud map, then the inclusion has a solution such that $x(t_0)=x_0$ (for all $x_0$) and for a sufficiently small time interval $[t_0,t_0+\varepsilon)$, $\varepsilon > 0$. Global existence, for all $t\in \R$ can be shown provided $\mathbb{F}$ does not allow ``blow-up'' ($\|x(t)\|\rightarrow\infty$ as $t\rightarrow t^*$ for a finite $t^*$).

The solutions to the switched system in \Cref{def:multi-robot-dynamics} provide the dynamical runs, treated as a differential inclusion defined on time \([0,T]\), as follows.
The \demph{dynamical run} is then the pair of trajectories \((\ONC(t), \EPI[t])\) for \(t \in [0,T]\), representing the evolution of the global ontic and epistemic states generated by the switched system in \Cref{def:multi-robot-dynamics}.
% }}}

\subsection{Computational Runs of Distributed Multi-Robot Systems} % {{{
\label{sec:runs-lcm}

We now describe the robot runs looking only at local epistemic computations. We use the fact that the global and local times associated to a time path correspond precisely to the clock of the environment and of the robots, respectively, in \Cref{def:robot-state-machine,def:env-state-machine}.

\begin{definition}[Execution] \label{def:exec}
  Let \(\robot \in \Pi\) be a robot, with an epistemic evolution specified by \Cref{def:robot-state-machine}. Let a pair \(\langle \epi_\robot, \obs_\robot \rangle\) be a \demph{robot configuration}. We say that an \demph{execution} of \(\robot\) is a sequence \(\{\langle \epi_\robot[t]_\robot,\obs_\robot[t]_\robot \rangle\}\) indexed by its local time \([t]_\robot \in N\), where \(\epi_\robot[t]_\robot \in S_\robot\), \(\obs_\robot[t]_\robot \in S_\obs\), and \(\fepi_\robot(\epi_\robot[t]_\robot, \obs_\robot[t]_\robot) = \epi_\robot[t+1]_\robot\).
\end{definition}

The robot execution can now be extended to a computational \remph{system run}, its robot system analog. For this, we define first the \remph{global robot configurations}, which will take the role of possible worlds in the epistemic frame, later described in \Cref{sec:task-tel}.

\begin{definition}[Global robot configuration]\label{def:global_config}
  A \demph{global robot configuration} is a tuple \(\cfg = \langle \{\epi_\robot\}_{\robot \in \Pi},\{\obs_\robot\}_{\robot \in \Pi}, [t] \rangle\) consisting of all local robot configurations at a global computation time \([t]\).
  \(\cfg_\robot\) is the restriction of a global configuration \(\cfg\) to the  robot \(\robot\).
\end{definition}

\begin{definition}[System run]\label{def:run}
  A \demph{system run} \(\run\) from time \(0\) to \(T\), determined by time path \(\tpath\), is a sequence of global robot configurations \(\{\cfg[t]\}_{[t] \in [0,T]}\), where, for any \(t \in [0,T]\):
  \begin{itemize}
    \item there exists a subset \(\varnothing \neq P_{\tpath}(t) \subseteq \Pi\) of participating robots, activated by the time path \(\tpath\) at global time \(t\).
    \item we have \[\cfg_\robot[t+1] = \\ \begin{cases}
        \fepi_\robot (\cfg[t]_\robot) & \textrm{ if } \robot \in P_{\tpath}(t) \\
        \cfg[t]_\robot & \textrm{ if } \robot \notin P_{\tpath}(t)
      \end{cases}\]
  \end{itemize}
\end{definition}

We denote the set of all possible runs $\mathcal{I}$, representing all the possible evolutions of the system solely based on the robot's behavior. An execution correspond to the individual epistemic evolution of a single robot, while a run represent the epistemic evolution of all the robots.
Note that this definition of run abstracts away the ontic (physical) updates in similar fashion to how the robot state machine represents solely the robots' computations. We effectively focus on the epistemic states, while the ontic evolution related constraints and interactions are encompassed by the environment.

% end Execution of Distributed Multi-Robot Systems }}}

\subsection{Correspondence of the Dynamical and Computational Perspectives} % {{{
\label{sec:correspondence}

We establish now the conditions on which the computational model of a multi-robot system from \Cref{sec:robot-as-computer} is a faithful abstraction of the switched hybrid dynamical system of \Cref{sec:robot-as-dynsys}. We show that the computational content of the system runs are the same in both dynamic (\Cref{sec:runs-dyn}) and computational (\Cref{sec:runs-lcm}) models. This justifies using the runs from computational models as the basis for the epistemic frames (\Cref{sec:task-tel}), as they capture all of the information available in the robots' decision process.

We rely on a system abstraction procedure that converts a hybrid multi-robot system \(\sysdyn\) into a computational system \(\syslcm\). We now define this abstraction and specify under which assumptions it preserves the computational information.

\begin{definition}[System abstraction]\label{def:abstraction}
  Given a hybrid system \(\sysdyn\), a \demph{system abstraction} \(\disc\) is a tuple of surjective functions \(\disc = (\disc_\EPI, \disc_\OBS, \disc_\CTL, \disc_\ONC)\) that map the global state spaces as follows.
  \begin{itemize}
    \item \(\disc_\EPI : \EPI \to S_{\epi}\), partitions the global epistemic state space into a finite set of discrete memory states;

    \item \(\disc_\OBS : \OBS \to S_{\obs}\), maps the global observations into a finite set of discrete observations, as per \Cref{rem:continuoustodiscrete};

    \item \(\disc_\CTL : \CTL \to S_{\ctl}\), maps the global control commands to a finite set of global actions;

    \item \(\disc_\ONC : \ONC \to S_{\env}\), partitions the global ontic state space into a finite set of environment states.
  \end{itemize}
\end{definition}

\begin{lemma}[Induced run abstraction]\label{lem:induced-run}
  A system abstraction \(\disc : \sysdyn \to \syslcm\) induces a function \(\disc^{*}: \runs_{\sysdyn} \to \runs_{\syslcm}\) that maps hybrid runs to computational runs.
\end{lemma}

\begin{proof}
  The function \(\disc^{*}\) is constructed by applying the system abstraction to every trajectory \((\ONC(t),\EPI[t]) \in \runs_{\sysdyn}\) at every discrete time step \([t] \in [0,T]\).
\end{proof}

The induced system abstraction need not be surjective, as some computational runs may not be physically possible. For instance, the computational model of a robot may allow for a state transition where its position updates faster than its speed limit. We filter those cases as follows.

\begin{assumption}[Physical realizability]\label{ass:realizability}
  Let \(\sysdyn\) be a hybrid system with protocol defined by its epistemic evolution \(\FEPI\) and control function \(CTL\). Let \(\syslcm\) be its computational abstraction. We assume that every run in \(\runs_{\syslcm}\) has a physical counterpart, i.e., \(\runs_{\syslcm} \subseteq \disc^{*}(\runs_{\sysdyn})\).
\end{assumption}

\begin{theorem}[Computational trace equivalence ]\label{thm:correspondence}
  Let \(\runs_{\sysdyn}\) be the set of all runs of a computational system \(\sysdyn\), and \(\runs_{\syslcm}\) the set of all runs of a hybrid system \(\syslcm\). If \(\syslcm\) is a faithful abstraction of a hybrid system \(\sysdyn\), then \(\disc^{*}(\runs_{\sysdyn}) = \runs_{\syslcm}\).
\end{theorem}

\begin{proof}
  We proceed by showing the mutual inclusion \(\disc^{*}(\runs_{\sysdyn}) \subseteq \runs_{\syslcm}\) and \(\runs_{\syslcm} \subseteq \disc^{*}(\runs_{\sysdyn})\).

  \textbf{\(\disc^{*}(\runs_{\sysdyn}) \subseteq \runs_{\syslcm}\)}. This follows from \Cref{lem:induced-run}, as the codomain of the induced system abstraction \(\disc^{*}\) is \(\runs_{\syslcm}\).

  \textbf{\(\runs_{\syslcm} \subseteq \disc^{*}(\runs_{\sysdyn})\)}. This holds directly from \Cref{ass:realizability}.
\end{proof}

This equivalence relies on \Cref{ass:realizability}, as we expect for the computation abstraction to be faithful to the physical dynamics. \Cref{thm:correspondence} formalizes the consequence of this assumption. 

% end Correspondence of dynamical and computational perspectives }}}

% end Schedulers and System Runs }}}

\section{Robot Tasks and Temporal Epistemic Logic} % {{{
\label{sec:task-tel}

% Introduction {{{
The remaining of this paper focuses on the computations of the distributed multi-robot systems described in \Cref{sec:robot-executions}, i.e., on their epistemic states and on their transitions. The epistemic evolution of robotic systems is described using epistemic frames that are build from the system runs of \Cref{def:run}:

\begin{definition}[Epistemic Frame] \label{def:epi-frame}
  Given a set of robots $\Pi$ and the set of all possible runs \(\runs\), an epistemic frame $\mathcal{F} = \langle \{\cfg\}_{\cfg \in \runs}, \{\sim_\robot\}_{\robot \in \Pi} \rangle$ is composed of a set of global robot configurations $\{\cfg\}_{\cfg \in \runs}$ (\Cref{def:global_config}) and a collection of indistinguishability relations $\{\sim_{\robot}\}_{\robot \in \Pi}$, where:
  \begin{itemize}
    \item Each global robot configuration, or \remph{point}, is constituted by a pair $(\run,[t])$ of system run $\run$ (\Cref{def:run}) and the corresponding global timestamp $[t]$ (\Cref{def:time}), which is relative to the run \(\run \in \runs\);
    \item the collection of indistinguishability relations $\sim_\robot$ indexed by robot $\robot \in \Pi$ are computed based on the equivalence of the local state: $(\run, [t]) \sim_\robot (\run', [t'])$ iff $e_{\robot}^{\cfg} = e_{\robot}^{\cfg'}$ where $\run[t]=\cfg$, i.e., if the epistemic state $e_\robot$ of $\robot$ in two different global robot configurations is the same.
  \end{itemize}
\end{definition}

The structure of the epistemic frame, and in particular its indistinguishability relations, depends on the assumptions of the robot model. For instance, if one considers that robots may not complete their \move\ steps (non-rigid model), uncertainty regarding their current locations will be added; visibility influences how much information can be gained (and discerned) during a \look\ step; asynchronicity allows indistinguisbility to range among points with the exact same global state but different global time, and so on. 

The epistemic frames, generated from the system runs of \Cref{def:run}, are used to reason about the knowledge of the multi-robot system throughout the robot tasks. Fittingly, we consider only the epistemic states of the system, as the ontic (physical) states are abstracted away in the transitions between computation steps.

In robot tasks, we assume that robots are able to sense a portion of the environment in which they are embedded, represented as a \remph{topological space}. We call the ``observable'' portion of the environment the \remph{observable space}, which we assume to be a \remph{compact} topological space, denoted by $(X, \tau)$. In particular, it is a compact normalized Euclidean space, i.e. $X \subseteq [0,1]^n$ with the standard Euclidean metric. To every open set $U$ of $\tau$ we associate a unique proposition $\xp(U)$. $\xp(X)$ is associated to the whole observation space $X$ and $\xp(\varnothing)$ to the empty space $\varnothing$. We refer to $\propspace$ as the set of \remph{space statements}. Space statements satisfy some important properties: $V \supseteq U$ implies that $\xp(V) \rightarrow \xp(U)$. Thus, $(\propspace, \rightarrow)$ is a poset, induced by the containment of open sets. Furthermore, since $(X,\tau)$ is closed under union $(\propspace, \rightarrow)$ is a join-lattice, allowing us to express, for example, information in a compact way: robot \remph{a} need only communicate to another robot \remph{b}, the largest region that contains all others, i.e., a single atom.

We say that a robot $\robot$ knows a formula $\psi$ at $(\run,[t])$ iff such a formula holds at all points that are indistinguishable from $(\run,[t])$ to $\robot$. In particular, because indistinguishability between points is defined on equivalence of robots' local states, the $\textbf{S5}$ epistemic logic\footnote{In an \textbf{S5} modal logic, the knowledge operator satisfies  factivity ($K_a \psi \rightarrow \psi$), positive introspection ($K_a \psi \rightarrow K_a K_a \psi$) and negative introspection ($\neg K_a \psi \rightarrow K_a \neg K_a \psi$).} is the natural candidate. This notion of knowledge can be extended to groups, such as \remph{distributed knowledge} $(\mathcal{D}_A \varphi)$  via the equivalence relation $\sim_{\mathcal{D}_A} := \bigcap_{\robot \in A} \sim_\robot$ consisting of the intersection of indistinguishabilities and representing the ``collective knowledge'' spread across all robots in $A$.

An epistemic model $\mathcal{M}$ is a tuple of an epistemic frame $\mathcal{F}$ and an evaluation function $\pi: \propspace \rightarrow 2 ^{\{\cfg\}_{\cfg \in \runs}}$ that assigns a truth value to atoms at each point of the model, where $\cfg$ is a global robot configuration, and $\propspace$ is the collection of space statements
The function $\pi$ simply maps a proposition $\xp(U)$ corresponding to a region $U$, to all the global configurations where $U$ holds.
In line with the runs and systems framework, we call our epistemic model \remph{interpreted system}, denoted by $\mathcal{I}$.

Note that, even though our logic is \emph{factive}, known facts about the real world are not necessarily true in the real world. As knowledge of robots depends on the robots \emph{observation}, factivity can only hold with respect to that observation. Consider as an example the formula $(\run, [t]) \models K_\robot position(\robot,x,y)$ denoting that robot $\robot$'s position at the time $t$ is at the point $(x,y)$ and lets assume that this was actually true in the actual world at the time robot observed its surrounding. This does not imply that it is \emph{still} true at any time $t' \neq t$ where $[t'] = [t]$. In fact our knowledge operator is only \textbf{S5} with respect to the epistemic observations, but not with respect to the \emph{real world}. Common knowledge experiences similar frictions, where attainability depends on the model and timing assumptions \cite{Fagin2003,HM90}, yet our modeling unveils the crucial dependency on \emph{observations}.
This is in line with the fact that we consider only the epistemic states of the system, abstracting away the ontic (physical) states in the transitions between computation steps.
% Introduction }}}

\subsection{A Temporal Epistemic Logic of Space} % {{{
\label{sec:logexplore}
The following grammar defines temporal epistemic logic formulas:

\begin{definition}[Language $\mathcal{L}_{space}$] \label{def:lang-exp}
    The language is defined by the following grammar:

    $$ \varphi ::= \xp(U) \: \vert \: \neg \varphi \: \vert \: (\varphi \wedge \varphi) \: \vert \: K_\robot \varphi \: \vert \: D_A \varphi \: \vert \: \Diamond \varphi$$

where $U \in \propspace$, $\robot \in \Pi$, and $A \subseteq \Pi$.  
\end{definition}

In this grammar, $\xp(U)$ is a space statement, $K_\robot \varphi$ represents ``$\robot$ knows $\varphi$'', $D_A \varphi$ stands for `` group \remph{A} has distributed knowledge of $\varphi$'', and $\Diamond \varphi$ reads ``eventually $\varphi$''.

\begin{definition}[Semantics of $\mathcal{L}_{space}$] \label{def:sem}
Let $\mathcal{I}$ be a set of runs (\Cref{def:epi-frame}), with its respective set of configuration-time pairs $\mathcal{W}$, an evaluation function $\pi: \propspace \rightarrow 2^{\mathcal{W}}$, a run $\run \in \mathcal{I}$, and a robot $\robot \in \Pi$, :

\begin{itemize}
    \item $\mathcal{I},(\run,[t]) \models \xp(U) \text{ iff } \xp(U) \in \pi(\run,[t])(\xp(U))$
    \item $\mathcal{I},(\run,[t]) \models \neg \varphi \text{ iff } \mathcal{I},(\run,[t]) \not\models \varphi $
    \item $\mathcal{I},(\run,[t]) \models (\varphi \wedge \psi) \text{ iff } \mathcal{I},(\run,[t]) \models \varphi \text{ and } \mathcal{I},(\run,[t]) \models \psi $
    \item $\mathcal{I},(\run,[t]) \models K_{\robot}\varphi \text{ iff } \mathcal{I}, (\run',[t]') \models \varphi \text{ for all } (\run',[t]')\in \mathcal{W} \text{ s.t. } 
        (\run,[t]) \underset{\robot}{\sim} (\run',[t]')$
    \item $\mathcal{I},(\run,[t]) \models D_A\varphi \text{ iff } \mathcal{I}, (\run',[t]') \models \varphi \text{ for all } (\run',[t]')\in \mathcal{W} \text{ s.t. } 
        (\run,[t]) \underset{{\mathcal{D}_A}}{\sim} (\run',[t]')$
    \item $\mathcal{I},(\run,[t]) \models \Diamond \varphi$ iff $\mathcal{I},(\run,[t]') \models \varphi $ for some $[t]' \geq [t]$.
\end{itemize}

We further set $\Box\varphi := \neg\Diamond\neg\varphi$ and $\varphi\vee\psi := \neg(\neg\varphi\wedge\neg\psi)$.
We also define mutual knowledge $E_\Pi \varphi$ as the conjunction of the knowledge of robots in $\Pi$: $E_\Pi \varphi := \bigwedge_{\robot \in \Pi} K_\robot \varphi$.

A formula $\varphi$ is valid in $\mathcal{I}$ if, for all $(\run,[t]) \in\mathcal{W}$,  $\mathcal{I} ,(\run,[t]) \models \varphi$. If $\varphi$ is valid in $\mathcal{I}$, we use the notation $\mathcal{I} \models \varphi$.
\end{definition}
% end Logic }}}

\subsection{Application to Exploration} % {{{
\label{sec:LCMexplore}

In this section we analyze the robot exploration task for a LCM model using the framework introduced in the previous section. Spatial atoms now denote that a certain region is \emph{explored}, i.e., some robot has observed it, without changing the semantics.

\begin{definition}[Simple robot exploration] \label{def:epi_cond}
    We say that a system is consistent with a simple exploration system if, for any $\robot \in \Pi$, and any region $U$, it satisfies:

  \begin{itemize}
    \item \textbf{Exploration agency:}  $\mathcal{I} \models \xp(U) \rightarrow  D_{\Pi} \xp(U)$, i.e., if a region $U$ has been explored, it has been explored by at least one robot \\
    \item \textbf{Exploration independence:}  $\mathcal{I} \models D_{A} \xp(U) \rightarrow \bigwedge_{\robot \in A}K_\robot \xp({V_\robot}) \textrm{ s. t. } \bigcup_{\robot \in A} V_\robot = U$, i.e., any space atom known distributively by robots in $A$, is given by the union of the space atoms known by robots in $A$ \\
    \item \textbf{Stable space statements:} $\mathcal{I} \models
                \bigwedge_{\xp(U) \in \mathcal{L}_{exp}^-} \xp(U) \rightarrow \Box \xp(U)$, i.e., once a region has been explored, it is forever explored.\\
    \item \textbf{Stable exploration knowledge (perfect recall):}  $\mathcal{I} \models K_{\robot} \xp(U) \rightarrow \Box K_\robot \xp(U)$, i.e., once a space atom is known by a robot, it will be known forever by that robot.
  \end{itemize}
\end{definition}

In the exploration task, if a robot knows that a certain region of space $\xp(U)$ has been explored, denoted by $K_\robot \xp(U)$, it implies that all subsets of that region are also known by that robot, i.e., $K_\robot \xp(V)$ with $V \subseteq U$ and $V$ open. This follows directly from the compactness of the exploration space.

\begin{definition}[Termination conditions]
\label{def:parterm}
  Given an interpreted system $\mathcal{I}$ consistent with simple robot exploration (\Cref{def:epi_cond}), with a set of robots $\Pi$, exploration space $X$, and its respective space statements $\propspace$, we say that $\mathcal{I}$ is consistent with:
  \begin{itemize}
    \item \textbf{Parallel termination}: iff $\mathcal{I} \models \Diamond D_\Pi \xp(X)$, i.e., no single robot needs to know that the whole space is explored;
    \item \textbf{Selfish termination}: iff $\mathcal{I} \models \Diamond E_\Pi \xp(X)$, i.e., each robot needs to know that the whole space is explored;
    \item \textbf{Cooperative termination}: iff $\mathcal{I} \models \Diamond E_\Pi \Diamond E_\Pi \xp(X)$, i.e., each robot needs to know that the whole space is explored and that other robots know that the whole space is explored.
  \end{itemize}
\end{definition}

Parallel termination can be characterized in terms of a very weak condition, called exploration liveness:

\begin{definition}[Exploration liveness]
\label{def:expllive}
    A run $\run$ satisfies exploration liveness iff there exists a collection of open sets, $\mathcal{U}$, such that $\bigcup_{U \in \mathcal{U}} U  = X$, and for any $U \in \mathcal{U}$, $\mathcal{I},\run \models \Diamond \xp(U)$.
    A model satisfies exploration liveness iff all of its runs satisfy exploration liveness.
\end{definition}

In other words, exploration liveness says that any point $x \in X$ of the exploration space $X$ will eventually be explored.

\begin{theorem}[Exploration Liveness Necessity]\label{thm:explivenec}
    Exploration liveness is necessary for simple exploration.
\end{theorem}
\begin{proof}
    \label{prf:explivenec}
    Note by contraposition that if there is a run $\run$, such that no open cover $\mathcal{U}$ of eventually explored regions exists. Naturally, we may consider $\mathcal{U}' = \{ U \; \vert \; \mathcal{I}, \run \models \Diamond \xp(U)\}$. Since $\mathcal{U}'$ is not an open cover of $X$, then $\bigcup_{U \in \mathcal{U}'} U \neq X$. Therefore, there is a point $x \in X, x \notin \bigcup_{U \in \mathcal{U}'} U$. Thus, by definition of $\mathcal{U}'$, $x$ is never explored, and consequently $X$ is not explored in $\run$.
\end{proof}

\begin{theorem}[Exploration Liveness Sufficiency]\label{thm:explivesuff}
    Let $\mathcal{I}$ be a model with a compact exploration space $X$. If $\mathcal{I}$ satisfies exploration liveness, then $\mathcal{I} \models \Diamond \xp(X)$.
\end{theorem}
\begin{proof}
\label{prf:explivesuff}
    Consider an arbitrary run $\run$ of $\mathcal{I}$, from \Cref{def:expllive}, there is an open cover $\mathcal{U}$ such that for any $U \in \mathcal{U}$, $\mathcal{I}, \run \models \Diamond \xp(U)$. Since $X$ is compact and $\mathcal{U}$ is an open cover of $X$ then there exists a finite subcover $\mathcal{U}' = \{ V_1, \ldots, V_k\} \subseteq \mathcal{U}$, such that $\bigcup_{i=1}^k V_i = X$. Also, since $\mathcal{I}, \run \models \Diamond \xp(V_i)$, then for any $i \in [1, k]$, then for each $V_i$ there is a time $[t_i]$ such that $\mathcal{I}, \run, [t_i] \models \xp(V_i)$. Simply consider $[t_{\textrm{max}}] = \mbox{}_{\textrm{max}} \{ [t_i]\}$. Since space statements are stable facts, then $\mathcal{I}, \run, [t_{\textrm{max}}] \models \bigwedge_{i= 1}^{k} \xp(V_i)$. Therefore $\mathcal{I}, \run, [t_{\textrm{max}}] \models \xp(X)$. Thus $\mathcal{I}, \run \models \Diamond \xp(X)$. Since $\run$ is an arbitrary run of $\mathcal{I}$, then it follows that $\mathcal{I} \models \Diamond \xp(X)$.
\end{proof}

\Cref{thm:explivenec}, \Cref{thm:explivesuff}, and \remph{exploration agency} from \Cref{def:epi_cond} imply that whenever $X$ is compact exploration liveness is necessary and sufficient for \remph{parallel termination}\footnote{Note that this is not necessarily true in non-compact spaces.}.

\begin{definition}[Exploration flooding]
    We say that the communication of a robot system is exploration flooding iff at each phase of the \LCM\ cycle, each robot tries to communicate the largest region of space it knows to be explored to the rest of the robots.
\end{definition}

\begin{lemma}\label{lem:eventualcomm}
Consider a model $\mathcal{I}$ with an exploration space $X$, from a $\LUMI$ robot system with full visibility and exploration flooding, and $\xp(U)$ an space statement, an arbitrary run $\run$ from $\mathcal{I}$, and an arbitrary robot $\robot \in \Pi$. If $\mathcal{I}, \run \models \Diamond K_\robot \xp(U)$, then $\mathcal{I}, \run \models \Diamond E_{\Pi} \xp(U)$.
\end{lemma}

\begin{proof}
\label{prf:eventualcomm}
    Since $\mathcal{I}, \run \models \Diamond K_\robot \xp(U)$, then there exists a time $[t_0]$ such that $\mathcal{I}, \run, [t_0] \models K_\robot \xp(U)$. Since robots have perfect recall, and space statements are stable, $\mathcal{I}, \run, [t] \models K_\robot \xp(U)$ for any $[t] \geq [t_0]$.

    Let $\robot' \in \Pi$ be an arbitrary robot. Since we don't consider crash failures, there is a time $[t]_{\robot'} \geq [t_0]$ when $\robot'$ is able to read the lights of robot $\robot$. Let $U'$ be the maximal region that $\robot$ knows to be explored at time $[t]_{\robot'}$. Since the robot system is exploration flooding, then robot $\robot$ is communicating $\xp(U')$ through its lights. Therefore, robot $\robot'$ learns $\xp(U')$, i.e. $K_{\robot'} \xp(U')$. However, recall that $U \subseteq U'$. Thus $\mathcal{I}, \run, [t]_{\robot'} \models K_{\robot'} \xp(U)$. 

    Since $\Pi$ is finite, then there is a $[t_{\textrm{max}}]$ such that $\mathcal{I}, \run, [t_{\textrm{max}}] K_{\robot''} \xp(U)$ for any robot $\robot'' \in \Pi$. Finally, $\mathcal{I}, \run \models \Diamond E_{\Pi} \xp(U)$.
\end{proof}

\begin{theorem}\label{thm:termequiv}
    Consider an interpreted system $\mathcal{I}$ with a compact exploration space $X$ such that $\mathcal{I}$ satisfies the exploration liveness condition, and such that the robot system is exploration flooding. Then $\mathcal{I}$ satisfies the cooperative termination property from \Cref{def:parterm}, namely, $\mathcal{I} \models \Diamond E_{\Pi} \Diamond E_{\Pi} \xp(X)$
\end{theorem}

\begin{proof}
\label{prf:termequiv}
    Since $X$ is compact and satisfies exploration liveness, then by \Cref{thm:explivenec}, then $\mathcal{I} \models \Diamond \xp(X)$. By exploration agency, then $\mathcal{I} \models \Diamond D_{\Pi} \xp(X)$. Let $\run$ be a fixed run of $\mathcal{I}$, then there exists a time $[t_0]$ such that $\mathcal{I}, (\run, [t_0]) \models D_{\Pi} \xp(X)$. Consider for each robot $\robot \in \Pi$, a region $U_\robot := \bigcup \{ U \subseteq X \; \vert \; \mathcal{I}, \run, [t_0] \models K_\robot \xp(U) \}$. Note that $K_\robot \xp(U_\robot)$, and $U_\robot$ is the maximal region that robot $\robot$ knows to be explored. From \remph{exploration independence}(~\Cref{def:epi_cond}), it follows that $\bigcup_{i=1}^k U_i = X$. 

    From \Cref{lem:eventualcomm}, it holds that $\mathcal{I}, \run \models \Diamond E_{\Pi} U_\robot$. Therefore, $\mathcal{I}, \run \models \Diamond E_{\Pi} \bigcup_{i=1}^n U_i$. This shows that $\mathcal{I}, \run \models \Diamond E_{\Pi} \xp(X)$.

    Let $[t_1]$ be such that $\mathcal{I}, \run, [t_1] \models E_{\Pi} \xp(X)$. In particular consider an arbitrary robot $\robot' \in \Pi$. Since $\mathcal{I}, \run, [t_1] \models K_\robot \xp(X)$, then by definition of knowledge, in any other run $\run'$ and any other time $[t']$ that is indistinguishable from $\run$ at time $[t_1]$ to $\robot'$, $\xp(X)$ holds. From \Cref{lem:eventualcomm}, it follows that $\mathcal{I}, \run' , [t'] \models \Diamond E_{\Pi} \xp(X)$. Thus, from the definition of the knowledge modality, $\mathcal{I}, \run, [t_1] \models K_{\robot'} \Diamond E_{\Pi} \xp(X)$. Since $\robot'$ is arbitrary, then $\mathcal{I}, \run, [t_1] \models E_{\Pi} \Diamond E_{\Pi} \xp(X)$. Therefore $\mathcal{I}, \run \models \Diamond E_{\Pi} \Diamond E_{\Pi} \xp(X)$. Therefore, $\mathcal{I} \models \Diamond E_{\Pi} \Diamond E_{\Pi} \xp(X)$.
\end{proof}

\begin{corollary} \label{cor:all-term-equiv}
    All termination conditions are equivalent for a $\LUMI$ robot system with exploration flooding, and exploration fairness. 
\end{corollary}

\begin{proof}
\label{prf:all-term-equiv}
    Note that \remph{Cooperative Termination} $\Rightarrow$ \remph{Selfish termination} $\Rightarrow$ \remph{Parallel Termination}. Finally note that \Cref{thm:termequiv} shows that \remph{Parallel Termination} $\Rightarrow$ \remph{Cooperative Termination}, which completes the equivalence cycle.
\end{proof}

Our result can be extended to cover arbitrary communication, nevertheless, a communication liveness condition would still be necessary.

\begin{definition}[Communication liveness] \label{def:commlive}
    We say that a model $\mathcal{I}$ with an exploration space $X$ satisfies the communication liveness condition if for any run $\run$ of $\mathcal{I}$, any time $[t]$, any robot $\robot$, and any space statement $p_u$, $\mathcal{I}, (\run, [t]) \models K_\robot \xp(U) \Rightarrow \mathcal{I}, \run \models \Diamond E_{\Pi} \xp(U)$.
\end{definition}

\begin{theorem}\label{thm:termequivcommlive}
    Consider a model $\mathcal{I}$ with a compact exploration space $X$ such that $\mathcal{I}$ satisfies the exploration liveness and communication liveness conditions, and such that the robot system is exploration flooding. Then $\mathcal{I}$ satisfies the cooperative termination property from \Cref{def:parterm}, namely, $\mathcal{I} \models \Diamond E_{\Pi} \Diamond E_{\Pi} \xp(X)$
\end{theorem}

\begin{proof}
\label{prf:termequivcommlive}
    Note that we can reuse the proof of \Cref{thm:termequiv}, but using \Cref{def:commlive} instead of \Cref{lem:eventualcomm}.
\end{proof}
% end exploration }}}

\subsection{Application to Surveillance} % {{{
\label{apx:surveillance}

In the surveillance task, we consider that an external agent might invade the robots space at a given time frame. If such an intruder happens to invade the space, then the robots should be able to detect it. In the following, we will show that the surveillance task can be formulated as a variant of the exploration task.

We will consider a normalized time frame for detection, represented by the closed interval $[0,1] \subset \mathbb{R}$. As with exploration, we consider a compact space $X$, which we will call our \remph{surveillance space}. Note that the product space of our surveillance space with the normalized time frame, $X \times [0,1]$ is a compact topological space with the product topology, which we will call the \remph{surveillance cylinder}.

Note that each region observed by the robots in the surveillance also induces a region in the surveillance cylinder by simply taking the time of the observations into account. For instance, if a robot explores a region $U$ of the surveillance space at a time $t$, then it already has explored $U\times t$ within the surveillance cylinder. Furthermore, if the intruder's speed is known to be bounded by a factor of $f$, then the exploration of $U$ at time $t_1$, implies that that at the next time $t_2$ at which the robot is active, $(U \setminus B_{f \cdot t} (\partial U)), t)_{t \in [t_1,t_2]}$ is also seen, where $B_{f \cdot t} (\partial U)$ is the ball of radius $f \cdot t$ around the boundary of $U$. This comes from the fact that if the intruder was not observed inside $U$ at time $t_1$, then the farthest that the intruder can reach at a time $t\in [t_1,t_2]$ is given by the ball of radius $f \cdot t$ around the boundary of $U$.

For the surveillance task, we will consider a set of propositional atoms consisting of space statements of the surveillance cylinder $X \times [0,1]$, namely $\propcyl$. Each of these space statements represent that a part of the surveillance cylinder was explored and no intruder was found. In addition to $\propcyl$, we also consider two additional propositional atoms, $\found$, and $\secure$. Thus, we can define a set of propositional atoms for surveillance, denoted by $\propsurv : = \propcyl \cup \{ \found, \secure \}$. Furthermore, we consider an epistemic temporal language for surveillance in the same way as we did for exploration, but with $\propsurv$ instead of $\propspace$. 

\begin{definition}[Language $\mathcal{L}_{surv}$]
\label{def:Lang-surv}
    The surveillance language is defined by the following grammar:

    $$ \varphi ::= p \: \vert \: \neg \varphi \: \vert \: (\varphi \wedge \varphi) \: \vert \: K_\robot \varphi \: \vert \: D_A \varphi \: \vert \: \Diamond \varphi$$

where $p \in \propsurv$, $\robot \in \Pi$, and $A \subseteq \Pi$.  
\end{definition}

The semantics of $\mathcal{L}_{surv}$ is analogous to the one of $\mathcal{L}_{space}$ in \Cref{def:sem}.
Note that the surveillance space is considered to be secure iff there does not exist a path $\xi:[0,1] \rightarrow X$ such that the trace of $\xi$. $tr(\xi) := \{ (\xi(t),t) \; \vert \; t \in [0,1]\}$ is not explored. A naive approach to secure the surveillance cylinder would be to consider the exploration on $X \times [0,1]$. However this is largely impractical, and often unfeasible, as it is equivalent to observing the space $X$ at all times during the time interval $[0,1]$. Instead, it is sufficient to explore a ``cut'' or section of the surveillance cylinder. For instance, if the surveillance space $X$ was explored completely at a fixed time time, i.e., if the set $X \times \{ t \}$ from the surveillance cylinder is explored, then the path of any intruder would intersect $X \times \{t\}$. More generally, if there is a manifold in the surveillance cylinder that has been observed, and can be expressed as the trace of a continuous function $f:X \rightarrow [0,1]$; $\trace(f):= \{ (x, f(x)) \; \vert \; x \in X\}$.

\begin{lemma}[Curve catching lemma]\label{lem:curvecatch}
    Let $\alpha : [0,1] \rightarrow X$ be an arbitrary curve in $X$, and $f:X \rightarrow [0,1]$ a continuous function. Then, the trace of $\alpha$, $\trace(\alpha):= \{ (\alpha(t), t) \; \vert \; t \in [0,1]\}$ and the trace of $f$, $\trace(f):= \{ (x,f(x)) \; \vert \; x \in X\}$ intersect.
\end{lemma}

\begin{proof}[Proof of \Cref{lem:curvecatch}] \label{prf:curvecatch}
Simply note that $f \circ \alpha: [0,1] \rightarrow [0,1]$ is a continuous function on $[0,1]$. Thus by the Brouwer's fixed-point theorem, there must exist a $t_0$ such that $f \circ \alpha (t_0) = t_0$. Let $\alpha(t_0) = y$. Thus $(y, t_0) \in \trace(\alpha)$ and by definition $(\alpha(t_0),f(\alpha(t_0))) \in \trace(f)$. Therefore $(y,t_0) \in \trace(\alpha)\cap \trace(f)$.
\end{proof}

Intuitively, this lemma shows that continuous manifolds on $X$ work as a ``catching net'' that is able to capture the path of any invader. Therefore, in addition to the conditions required for simple exploration (~\Cref{def:epi_cond}), we also consider the following conditions:

\begin{definition}[Found and secure conditions]\label{def:survcond}
\mbox{}\\
  \textbf{Detection Agency} Consider an arbitrary run $\run$ in a robot system, and an intruder path $\alpha: [0,1] \rightarrow X$. We say that $\found$ holds at a time $t$ iff either there exists a robot $p_i$ whose observation region at time $t$ intersects $\alpha$, or there is a time $t'< t$ such that $\found$ holds at time $t'$.\\
  \textbf{Observation Correctness} Consider an arbitrary run $\run$, and an intruder path $\alpha: [0,1] \rightarrow X$. Let $U$ be a region of the surveillance cylinder $X \times [0,1]$ such that $U \cap \trace(\alpha) \neq \varnothing$. Then $\neg \xp(U)$.\\
  \textbf{Securement} Let $f:X \rightarrow[0,1]$ be any continuous function. Let $U$ be a region of the surveillance cylinder $X \times [0,1]$, if $ \trace(f) \subseteq U$, then $\xp(U) \rightarrow \secure$.\\
  \textbf{Completion condition} $ \mathcal{I} \models (\secure \implies \Box \neg \found) \wedge (\found \implies \Box \neg \secure)$.
\end{definition}

Detection agency means that the $\found$ atom only holds if an intruder has already been detected by a robot. Observation correctness means that a region in the surveillance cylinder is not secure if an intruder is moving through a space-time point in the region, i.e., exploration regions in this context refer only to regions where no intruder has been found. Securement captures the idea from \Cref{lem:curvecatch} that shows that it is sufficient to explore the trace of a continuous manifold in $X$.
Completion condition states that if part of the surveillance cylinder was explored then the intruder was not found (and never will), and if the intruder was found then the surveillance cylinder was not secure (and never will be).
Also note that if the trace of the manifold has even a single discontinuity point, then this may be enough for an intruder to evade detection. 

We provide the following definition for the surveillance robot task (without termination):

\begin{definition}[Surveillance]\label{def:surv}
Consider a robot system and a model $\mathcal{I}$ consistent with \Cref{def:epi_cond} with respect to its surveillance cylinder, that satisfies \Cref{def:survcond}. Let $\run$ be a run of $\mathcal{I}$. We say that the model solves the surveillance task if $\mathcal{I}, \run \models \Diamond \found \vee \Diamond \secure$.
\end{definition}
Note that by Completion condition (\Cref{def:survcond}) $\secure$ and $\found$ atoms are mutually exclusive in a run.

As with the simple exploration task, we can consider multiple termination conditions for the surveillance task.
\begin{definition}[Termination for surveillance]\label{def:termsurv}
\mbox{}\\
    \textbf{Parallel Termination} $\mathcal{I}, \run \models \Diamond D_{\Pi} \found \vee \Diamond D_{\Pi} \secure$\\
    \textbf{Selfish Termination} $\mathcal{I}, \run \models \Diamond E_{\Pi} \found \vee \Diamond E_{\Pi} \secure$\\
    \textbf{Cooperative Termination} $\mathcal{I}, \run \models \Diamond E_{\Pi} \Diamond E_{\Pi} \found \vee \Diamond E_{\Pi} \Diamond E_{\Pi} \secure$
\end{definition}

Since \Cref{lem:curvecatch} shows that exploring a continuous manifold of the surveillance cylinder is enough to detect an intruder, it is also enough to fully secure the space, thus we can replace the exploration liveness condition to a manifold instance.

\begin{definition}[Manifold liveness] \label{def:manilive}
    We say that run $\run$ of a robot system has the manifold liveness condition iff  there exists a continuous function $f_\run X \rightarrow [0,1]$ such that $\run$ satisfies the exploration liveness condition for $\trace(f_\run)$.
\end{definition}

\begin{theorem}[Manifold liveness sufficiency]\label{thm:mansuff}
    Consider a robot system that satisfies the simple exploration conditions, \Cref{def:epi_cond}, for the surveillance cylinder $X \times [0,1]$, the surveillance conditions, \Cref{def:survcond}, and the manifold surveillance liveness condition, \Cref{def:manilive}. Then the surveillance task is solved at $\run$.
\end{theorem}

\begin{proof}[Proof of \Cref{thm:mansuff}] \label{prf:mansuff}
    Follows from \Cref{lem:curvecatch}, and \Cref{thm:explivesuff}.
\end{proof}

\begin{theorem} \label{thm:surv-coop}
    Consider a model $\mathcal{I}$ with a compact surveillance space $X$ that satisfies the communication liveness, the manifold liveness condition, the simple exploration conditions for the surveillance cone $X\times[0,1]$, the surveillance conditions, and the exploration flooding with respect to the surveillance cylinder. Then $\mathcal{I}$ satisfies the cooperative termination property.
\end{theorem}

\begin{proof}[Proof of \Cref{thm:surv-coop}] \label{prf:surv-coop}
Let $\run$ be a run of $\mathcal{I}$, from \Cref{thm:mansuff}, $\mathcal{I}, \run \models \Diamond \found \vee \Diamond \secure$.
Recall that by Completion condition (\Cref{def:survcond}) $\secure$ and $\found$ are mutually exclusive in a run.
First, let us assume that $\mathcal{I}, \run \models \Diamond \found$. From \Cref{def:survcond}, there must exist a robot $\robot_i$ such that finds an intruder at a time $t_i$. Since $\robot_i$ detects the intruder, it holds that $\mathcal{I}, \run, t_i \models K_{\robot_i} \found$. From \Cref{def:commlive}, it holds that $\mathcal{I} \run \models \Diamond E_{\Pi} \found$. In particular, in all runs $\run'$ that are indistinguishable to $\robot_i$, $\mathcal{I}, \run' \models \Diamond E_{\Pi} \found$. Therefore, $\mathcal{I}, \run, t_i \models K_{\robot_i} \Diamond E_{\Pi} \found$. Since $\mathcal{I}, \run \models \Diamond E_{\Pi} \found$, then we can repeat the previous argument for any other robot. Therefore $\mathcal{I}, \run \models \Diamond E_{\Pi} \Diamond E_{\Pi} \found$.

Now, assume that $\mathcal{I}, \run \models \Diamond\secure$ holds. Then, it follows that no intruder was found. From \Cref{def:manilive}, it holds that there exists a continuous manifold $f_\run:X \rightarrow [0,1]$, such that $\mathcal{I}, \run \models \Diamond p_{\trace(f_\run)}$. Note that the conditions for \Cref{thm:termequivcommlive} hold, therefore, $\mathcal{I}, \run \models \Diamond E_{\Pi} \Diamond E_{\Pi} p_{\trace(f_\run)}$. From the securement condition from \Cref{def:survcond}, it follows that $\mathcal{I}, \run \models \Diamond E_{\Pi} \Diamond E_{\Pi} \secure$.
\end{proof}

\begin{corollary}
    In particular for the $\LUMI$ model with full visibility, if the robot system satisfies the manifold liveness, then it solves the surveillance task with cooperative termination
\end{corollary}
% surveillance }}}

\subsection{Application to Approximate Gathering} % {{{
\label{sec:gathering}

In the approximate gathering task, robots needs to ensure that their position and all other robots positions' are eventually contained in a specific open set, called the \remph{rendezvous region}.
For this epistemic modeling we need to extend the language of \Cref{def:lang-exp} with additional atoms:

\begin{itemize}
    \item \textbf{Positional atoms} $pos_\robot(x) \in \propobs \subseteq X$ (i.e., a finite subset of the domain of the observation function $\gobs$ in \Cref{def:robot-ode}) with $|\propobs| \in \N$;

    \item \textbf{Initial position atoms} $init(pos_\robot(x)) \in \propobs$, meaning the starting position of robot $\robot$;

    \item \textbf{Inclusion atoms} $(x\in U) \in \proptopo$ representing the inclusion of points in open sets in the exploration space $\in_{x,U}$;
\end{itemize}

We define the set of \textit{Rendezvous region} $\proprdvz \in 2^X$, which are a a \remph{finite} set of regions of space such that $\forall U, V \in \proprdvz: U \cap V = \emptyset$ and $|\proprdvz| \in \N$. $\vec{x}\in\propspace^\Pi$ denotes a vector with a position for each robot. 
The following grammar defines temporal epistemic logic formulas for approximate gathering.

\begin{definition}[Language $\mathcal{L}_{gather}$]
\label{def:lang_approx_gathering}
  The gathering language is defined by the following grammar: \[\varphi ::= p \: \vert \: \neg \varphi \: \vert \: (\varphi \wedge \varphi) \: \vert \: K_r \varphi \: \vert \: D_A \varphi \: \vert \: \Diamond \varphi\] where $p \in \proptopo \cup \propobs$ , $\robot \in \Pi$, and $A \subseteq \Pi$. When obvious from the context, we will omit the subscript $X$.
\end{definition}

The semantics of $\mathcal{L}_{gather}$ is analogous to the one of $\mathcal{L}_{exp}$ in \Cref{def:sem}, except for the newly introduced atoms:

\begin{itemize}
  \item $\mathcal{I},(\run,[t]) \models pos_\robot(x) \text{ iff } pos_\robot(x) \in \pi(\run,[t])(pos_\robot(x))$

  \item $\mathcal{I},(\run,[t]) \models init(pos_\robot(x)) \text{ iff } pos_\robot(x) \in \pi(\run,0)(pos_\robot(x))$

  \item $\mathcal{I},(\run,[t]) \models (x \in U) \text{ iff } (x \in U) \in \pi(\run,[t])(x \in U)$
\end{itemize}

Positional atoms $pos_\robot(x)$ mean that robot $\robot$ is at actual position $\onc$ discretized to $x$. As $\gobs$ is surjective, it is possible that $\gobs_\robot(\onc) = \gobs_\robot(\onc')$ even though $\onc \neq \onc'$, meaning two different actual robot positions result in the same position for that robot.
We assume that robots know their own position: $pos_\robot(x) \leftrightarrow K_\robot pos_\robot(x)$.

The approximate gathering is solved if all robots' positions are contained in the \remph{same} rendezvous region forever, assuming each robot starts in one rendezvous region. To do so, they need to agree on one rendezvous region and move here. They are allowed to change the rendezvous region finitely often, but eventually need to stick to one forever.
We require it to be impossible for any robot to intersect two rendezvous regions at the same time, i.e., $pos_r(x)\wedge(x\in U)\implies(\bigwedge_{U\neq V\in\proprdvz}\neg(x\in V))$.

\begin{definition}[Approximate gathering] \label{def:term_approx_gathering}
  System $\mathcal{I}$ is consistent with approximate gathering iff the following two conditions hold:

  \begin{itemize}
    \item Eventual Gathering: There is a rendezvous region where all robots eventually meet forever 
        $$\mathcal{I} \models \bigvee_{U \in \proprdvz}
        \bigwedge_{\robot\in\Pi} \Diamond \Box \bigvee_{x \in \propobs} pos_\robot(x) \wedge (x\in U)$$
    \item Starting Validity: If all robots are already contained in a rendezvous region initially, that is also the final rendezvous region 
        $$\mathcal{I} \models \bigwedge_{\substack{U \in \proprdvz, \vec{x} \in \propobs^\Pi}} \bigl((\bigwedge_{r\in\Pi} init(pos_\robot(x_\robot))) \wedge (x_\robot\in U) \bigr)\rightarrow \bigl(\Diamond \Box \bigvee_{\vec{y} \in \propobs^\Pi} ( \bigwedge_{\robot\in\Pi} pos_\robot(y_\robot)) \wedge (y_\robot\in U) \bigr).$$
    \end{itemize}
\end{definition}

These conditions strongly resemble the stabilizing agreement conditions, epistemically characterized in \cite{crgfessli2024}. It is defined by the following formulas, where $\mathcal{V}$ is the set of values to choose from, the atom $choose_a(v)$ denotes that agent $a$ chooses value $v$. $decide_a(v)$ is defined as $\Box choose_a(v)$, where $a,b$ are \remph{agents}, i.e., robots with no dynamical constraints:

\begin{definition}[Stabilizing agreement \cite{crgfessli2024}] \label{def:Ef_ST}
  System $\mathcal{I}$ is consistent with stabilizing agreement if \remph{(Agreement)} and \remph{(Validity)} hold.
  \begin{itemize}
    \item Agreement: There is a value such that every agent decides on that value 
       $$ \mathcal{I} \models \bigvee_{v \in \mathcal{V}} \bigwedge_{a \in \Pi} \Diamond decide_a(v)$$
    \item Validity: An agent can only choose a known initial value of some agent 
      $$\mathcal{I} \models \bigwedge_{v \in \mathcal{V}} \bigl(choose_a(v) \rightarrow   K_a \bigvee_{b \in \Pi} init_b (v) \bigr)$$
\end{itemize}
\end{definition}

Observe that covering a rendezvous region can be seen as that robot choosing that region. Rendezvous regions are the objects the robots choose from. As each robot initially sits in exactly one region, that region can be seen as the problem input. The last obstacle is reducing the \remph{strong validity}, as its called in the distributed computing literature (see for example \cite{csrjv23}), to the \remph{weak validity} of approximate gathering.

\begin{theorem}\label{thm:stabagapproxgather}
  Consider a model $\mathcal{I}$ with a space $X$, together with the language defined \Cref{def:lang_approx_gathering}, an arbitrary run $\run$ from $\mathcal{I}$. If $\mathcal{I}$ is consistent with stabilizing agreement, then it is also consistent with approximate gathering.

  Consider setting the set of possible choices $\mathcal{V}$ to be the set of possible rendezvous regions $\proprdvz$, setting $choose_\robot(U) \leftrightarrow \bigvee_{x\in\propobs} pos_\robot(x) \wedge (x\in U)$ and $decide_\robot(U) = \Box choose_\robot(U)$, and setting $init_\robot(U) \leftrightarrow \bigvee_{x\in\propobs} init(pos_\robot(x))\wedge(x\in U)$.
\end{theorem}

\begin{proof}
\label{prf:stabagapproxgather}
  We start by substituting for $choose$ and $init$ as defined in \Cref{def:Ef_ST}. For (Agreement) we immediately arrive at (Eventual Gathering):

  \begin{equation}\label{eq:proofagreement}
    \mathcal{I} \models \bigvee_{U\in\proprdvz} \bigwedge_{\robot\in\Pi} \Diamond \Box \bigvee_{x\in\propobs} pos_\robot(x) \wedge (x\in U).
  \end{equation}

  For (Starting Validity) we need to provide more reasoning, consider the following, where we already prepended the validity with a conjunction over all robots:

  \begin{equation}\label{eq:proofassumption}
    \begin{split}
      \mathcal{I} \models \bigwedge_{\substack{\robot\in\Pi,\\ U\in\proprdvz}}
          \bigl((&\bigvee_{x\in\propobs} pos_\robot(x) \wedge (x\in U)) \rightarrow \\
          (K_\robot &\bigvee_{\robot' \in \Pi} \bigvee_{x\in\propobs} init(pos_{\robot'}(x))\wedge(x\in U)) \bigr).
    \end{split}
  \end{equation}

  As we need to prove an implication, we assume its antecedent and derive the consequent. Assuming
  \begin{equation}\label{eq:proofantecedent}
    \bigwedge_{\robot\in\Pi} init(pos_\robot(x_\robot)) \wedge (x_\robot\in U),
  \end{equation}
  we want to derive
  \begin{equation}\label{eq:proofconsequent}
     \bigl( \Diamond \Box \bigvee_{\vec{y} \in \propobs^\Pi} ( \bigwedge_{\robot\in\Pi} pos_\robot(y_\robot) ) \wedge ( y_\robot \in U )\bigr),
  \end{equation}
  for any $U\in\proprdvz$ and any $\vec{x}\in\propobs^\Pi$. Choose any $U$, $\vec{x}$ which satisfies \Cref{eq:proofantecedent}, meaning a starting position vector that intersects over $U$. Note that there is \remph{only one} such $U$ as initial rendezvous regions are unique.

  Observe that, as robots know their own initial positions, \Cref{eq:proofantecedent} implies the antecedent in \Cref{eq:proofassumption} for all robots $\robot$:
  \begin{equation}\label{eq:proofstep3}
  \begin{split}
               &init(pos_\robot(x_\robot)) \wedge (x_\robot \in U) \implies\\
      K_\robot &init(pos_\robot(x_\robot)) \wedge (x_\robot \in U) \implies\\
      K_\robot \bigvee_{\robot'\in\Pi}   &init(pos_{\robot'}(x_{\robot'})) \wedge (x_{\robot'} \in U) \implies\\
     (K_\robot \bigvee_{\robot'\in\Pi} \bigvee_{y\in\propobs} &init(pos_{\robot'}(y_{\robot'}))\wedge(y_{\robot'} \in U)).
  \end{split}
  \end{equation}
  Therefore the consequent in \Cref{eq:proofassumption} is true for only one choice of $U$ (the same that satisfies \Cref{eq:proofantecedent}) and false for all other $V \in \proprdvz$.

  By \Cref{eq:proofagreement} there exists some $V$, $t$ such that
  \begin{equation}\label{eq:proofstep1}
      \mathcal{I}, (\run,t) \models \bigvee_{x\in\propobs} pos_\robot(x) \wedge (x\in V)
  \end{equation}
  holds for all $t'>t$ and for all robots $r$. We see that \Cref{eq:proofstep1} is precisely the antecedent of \Cref{eq:proofassumption}, and, as \Cref{eq:proofassumption} is a validity, it is satisfied by the fixed choice for $V$, $t$. Since the consequent in \Cref{eq:proofassumption} \remph{is true only} for $U$ as derived in \Cref{eq:proofstep3} (but \Cref{eq:proofassumption} is a validity), \Cref{eq:proofstep1} has to be false for any $W \neq U$. But by the previous reasoning we know that \Cref{eq:proofstep1} is true for at least one $V$ for all $t'>t$, implying that $U = V$ for all $t'>t$. We can now derive \Cref{eq:proofconsequent} from \Cref{eq:proofstep1}:
  \begin{equation}
  \begin{split}
      \bigvee_{x\in\propobs} &pos_\robot(x) \wedge (x\in U) \\
      \bigvee_{\vec{y} \in \propobs^\Pi} ( \bigwedge_{\robot\in\Pi} &pos_\robot(\vec{y}_r) ) \wedge (x \in \vec{y}_r )\bigr) \\
      \bigl( \Diamond \Box \bigvee_{\vec{y} \in \propobs^\Pi} ( \bigwedge_{\robot\in\Pi} &pos_\robot(\vec{y}_r)) \wedge (x \in \vec{y}_r )\bigr),
  \end{split}
  \end{equation}
  meaning the robots gather at the initially intersected rendezvous region forever.
\end{proof}
% end gathering }}}

% robot tasks }}}

\section{Conclusions and future work} % {{{
\label{sec:concl}
In this paper, we developed a framework that seamlessly integrates control theory, distributed systems and temporal epistemic logic in the context of multi-robot systems, with applications to exploration, gathering and surveillance tasks. We present, at the same time, new complementary perspectives for modeling distributed robot systems as dynamical systems and as state machines. The novel separation of a robot state into ontic and epistemic parts is exploited by developing a common understanding of scheduler that ties their local executions to a global time. We extend this first bridge between control theory and distributed computing by also providing access to the epistemic reasoning present in distributed computing, now adapted for multi-robot tasks and their particularities, namely the fact that robots exist in some (physical) space. This reasoning arena takes the form of an epistemic model capable of expressing the knowledge of the system at their computation operations. A temporal epistemic logic of space is employed and its usefulness made evident with applications to concrete exploration and gathering tasks, for both of which sufficient epistemic conditions are derived.

This work has multiple implications and holds the promise of future developments in exchanging tools and results across fields.  One of the most striking implications is that the classical luminous robot model \cite{DasFlocchiniPrencipeSantoroYamashita16} with full visibility is very strong in a distributed computing sense, as confirmed by the existence of a whole research field on establishing and maintaining communication in multi robot systems.  Some, such as \cite{BB23,BGB23} have leveraged epistemic planning and action models in order to keep a live communication.  As we have identified a sufficient condition for approximate gathering, we can directly judge the correctness of an action model, and therefore the correctness of a protocol, by its ability to reach a given specification in epistemic terms. 

The results of this paper have already been extended for the task of surveillance, as can be found in \Cref{apx:surveillance}, where a group of robots must monitor a region of space to be absent of intruders.
Other interesting extensions include (but are not limited to) considering other robot tasks and their similarities to distributed tasks, as well as extending the logic framework with deontic modalities expressing  robots' permissions and obligations during protocol execution given their current local knowledge. Furthermore, our framework is capable of integrating the geometrical interpretation \cite{Roman} of epistemic frames and the recent advances \cite{RiessGhrist22} on sheaf theoretic and homological aspects of task solvability, however this is out of scope for this paper, and could be the focal interest of follow-up work.
% end Conclusion }}}

% References {{{
\bibliography{lit}
% }}}

\appendix % {{{
\section{Taxonomy of \textsc{look-compute-move}} \label{apx:lcm} % {{{

There are several possible choices regarding the capacities of robots. The following \Cref{tab:glossary_lcm} summarizes the most relevant variations for robots under the \textsc{look-compute-move} abstraction. They are chosen in order to make a model mode or less powerful, with additional assumptions on their communication, observation and movement capacities.

\begin{table}[ht!]
  \caption{Partial glossary of \LCM terminology.} \label{tab:glossary_lcm}
  \begin{tabular*}{\columnwidth}{llp{7cm}}
    \toprule
    Property & Variant & Description  \\
    \midrule
    Memory & \textbf{Oblivious} & No memory, all information is lost between rounds \\
           & \textbf{Luminous} & Lights can be used to preserve information across rounds \\
    \midrule
    Communication & \textbf{Silent} & No communication between robots \\
                  & \textbf{Luminous} & Information passed indirectly by retrieving lights during \look \\
%                  Robots are equipped with lights for indirect communication \\
    \midrule
    Visibility & \textbf{Myopic} & \look retrieves a partial snapshot parametrized by view distance\\
               & \textbf{Unlimited visibility} & \look retrieves a snapshot of the entire environment \\
    \midrule
    Identity & \textbf{Anonymous} & No unique identifier \\
             & \textbf{Indistinguishable} & No external identifier \\
    \midrule
    Movement & \textbf{Rigid} & \move is always executed fully \\
             & \textbf{Non-rigid} & \move can be interrupted, parametrized by a minimum distance \\
    \midrule
    Dynamics & \textbf{Holonomic} & \move can be any arbitrary destination \\
             & \textbf{Non-holonomic} & \move has to respect dynamic constraints, e.g. Dubins' car \\
    \midrule
    Orientation & \textbf{Oriented} & There is a global coordinate system that all can access \\
                & \textbf{Disoriented} & Only different local coordinate systems \\
                & \textbf{Chiral} & Different local coordinate systems, but with common north \\
    \midrule
    Synchronicity & \textbf{ASYNC} & Scheduler, no shared notion of time \\
              & \textbf{k-ASYNC} & Scheduler, the robots are at most \(k\) rounds apart \\
              & \textbf{SSYNC} & Scheduler, a subset of the robots activated at each round \\
              & \textbf{FSYNC} & Scheduler, all robots are activated at each round \\
    \bottomrule
  \end{tabular*}
\end{table}  % tab glossary lcm

% end Taxonomy }}}
% end Appendix }}}
\end{document}